\documentclass[sts]{imsart}
\input{commands.tex}

\begin{document}
\begin{frontmatter}
\title{Near-linear time approximation algorithms for optimal transport via Sinkhorn iteration}
\runtitle{Near-linear time optimal transport}
\author{\fnms{Jason}~\snm{Altschuler}\thanksref{t1}\ead[label=jason]{jasonalt@mit.edu}},
\author{\fnms{Jonathan}~\snm{Weed}\thanksref{t1}\ead[label=jon]{jweed@mit.edu}},
\and
\author{\fnms{Philippe}~\snm{Rigollet}\thanksref{t2}\ead[label=rigollet]{rigollet@math.mit.edu}}

\affiliation{Massachusetts Institute of Technology}
\thankstext{t1}{This work was supported in part by NSF Graduate Research Fellowship DGE-1122374.}
\thankstext{t2}{This work was supported in part by NSF CAREER DMS-1541099, NSF DMS-1541100, DARPA W911NF-16-1-0551, ONR N00014-17-1-2147 and a grant from the MIT NEC Corporation.}

\begin{abstract} 
Computing optimal transport distances such as the earth mover's distance is a fundamental problem in machine learning, statistics, and computer vision. Despite the recent introduction of several algorithms with good empirical performance, it is unknown whether general optimal transport distances can be approximated in near-linear time. This paper demonstrates that this ambitious goal is in fact achieved by Cuturi's \emph{Sinkhorn Distances}. This result relies on a new analysis of Sinkhorn iteration, which also directly suggests a new greedy coordinate descent algorithm, $\greedy${}, with the same theoretical guarantees. Numerical simulations  illustrate that $\greedy$ significantly outperforms the classical $\sinkhorn$ algorithm in practice.
\end{abstract}

\medskip
\begin{center}\textit{Dedicated to the memory of Michael B. Cohen}\end{center}
\vspace{-10pt}
\end{frontmatter}

\section{Introduction}\label{sec:intro}
Computing distances between probability measures on metric spaces, or more generally between point clouds, plays an increasingly preponderant role in machine learning~\cite{SanLin11, MueJaa15, LloGha15,JitSzaChw16,ArjChiBot17},  statistics~\cite{FlaCutCou16,PanZem16,SzeRiz04,BigGouKle17} and computer vision~\cite{RubTomGui00,BonPanPar11, SolGoePey15}. A prominent example of such distances is the \emph{earth mover's distance} introduced in~\cite{WerPelRos85} (see also \cite{RubTomGui00}), which is a special case of Wasserstein distance, or optimal transport (OT) distance~\cite{Vil09}. 

While OT distances exhibit a unique ability to capture geometric features of the objects at hand, they suffer from a heavy computational cost that had been prohibitive in large scale applications until the recent introduction to the machine learning community of \emph{Sinkhorn Distances} by Cuturi~\cite{Cut13}. Combined with other numerical tricks, these recent advances have enabled the treatment of large point clouds in computer graphics such as triangle meshes~\cite{SolGoePey15} and high-resolution neuroimaging data~\cite{GraPeyCut15}. Sinkhorn Distances rely on the idea of \emph{entropic penalization}, which has been implemented in similar problems at least since Schr\"odinger~\cite{Sch31, Leo14}. This powerful idea  has been successfully applied to a variety of  contexts not only as a statistical tool for model selection~\cite{JudRigTsy08, RigTsy11,RigTsy12} and online learning~\cite{CesLug06}, but also as an optimization gadget in first-order optimization methods such as mirror descent and proximal methods~\cite{Bub15}.

\medskip

{\bf Related work.} Computing an OT distance amounts to solving the following linear system:
\begin{equation}
\label{EQ:OT}
\min_{P \in \cU_{r,c}} \langle P, C \rangle\,, \qquad \quad  \cU_{r,c}:=\big\{P \in \R^{n\times n}_+\,:\, P\bone=r\,, P^\top \bone=c\big\}\,,
\end{equation}
where  $\bone$ is the all-ones vector in $\R^n$, $C \in \R_+^{n \times n}$ is a given \emph{cost matrix}, and $r \in \R^n, c \in \R^n$ are given vectors with positive entries that sum to one. Typically $C$ is a matrix containing pairwise distances (and is thus dense), but in this paper we allow $C$ to be an arbitrary non-negative dense matrix with bounded entries since our results are more general. For brevity, this paper focuses on square matrices $C$ and $P$, since extensions to the rectangular case are straightforward.

This paper is at the intersection of two lines of research: a theoretical one that aims at finding (near) linear time approximation algorithms for simple problems that are already known to run in polynomial time and a practical one that pursues fast algorithms for solving optimal transport approximately for large datasets.

Noticing that~\eqref{EQ:OT} is a linear program with $O(n)$ linear constraints and certain graphical structure, one can use the recent Lee-Sidford linear solver to find a solution in time $\tildo(n^{2.5})$~\cite{LeeSid14}, improving over the previous standard of $O(n^{3.5})$ \cite{Ren88}. While  no practical implementation of the Lee-Sidford algorithm is known, it provides a theoretical benchmark for our methods. Their result is part of a long line of work initiated by the seminal paper of Spielman and Teng~\cite{SpiTen04SDD} on solving linear systems of equations, which has provided a building block for near-linear time approximation algorithms in a variety of combinatorially structured linear problems. A separate line of work has focused on obtaining faster algorithms for~\eqref{EQ:OT} by imposing additional assumptions. For instance, ~\cite{AgaSha14} obtain approximations to~\eqref{EQ:OT} when the cost matrix $C$ arises from a metric, but their running times are not truly near-linear.~\cite{ShaAga12,AndNikOna14} develop even faster algorithms for~\eqref{EQ:OT}, but require $C$ to arise from a low-dimensional $\ell_p$ metric.

Practical algorithms for computing OT distances include Orlin's algorithm for the \emph{Uncapacitated Minimum Cost Flow} problem via a standard reduction. Like interior point methods, it has a provable complexity of $O(n^3\log n)$. This dependence on the dimension is also observed in practice, thereby preventing large-scale applications. To overcome the limitations of such general solvers, various ideas ranging from graph sparsification~\cite{PelWer09} to metric embedding~\cite{IndTha03,GraDar04, ShiJac08} have been proposed over the years to deal with particular cases of OT distance. 

Our work complements both  lines of work, theoretical and practical, by providing the first near-linear time guarantee to approximate~\eqref{EQ:OT} for general non-negative cost matrices. Moreover we show that this performance is achieved by algorithms that are also very efficient in practice.
Central to our contribution are recent developments of scalable methods for general OT that leverage the idea of entropic regularization~\cite{Cut13, BenCarCut15,  GenCutPey16}. However, the apparent practical efficacy of these approaches came without theoretical guarantees. In particular, showing that this regularization yields an algorithm to compute or approximate general OT distances in time nearly linear in the input size $n^2$ was an open question before this work.

\medskip

{\bf Our contribution.} 
The contribution of this paper is twofold. First we demonstrate that, with an appropriate choice of parameters, the algorithm for Sinkhorn Distances introduced in~\cite{Cut13} is in fact a \emph{near-linear time} approximation algorithm for computing OT distances between discrete measures. This is the first proof that such near-linear time results are achievable for optimal transport. We also provide previously unavailable guidance for parameter tuning in this algorithm. Core to our work is a new and arguably more natural analysis of the Sinkhorn iteration algorithm, which we show converges in a number of iterations independent of the dimension $n$ of the matrix to balance. In particular, this analysis directly suggests a greedy variant of Sinkhorn iteration that also provably runs in near-linear time and significantly outperforms the classical algorithm in practice. Finally, while most approximation algorithms output an approximation of the optimum \emph{value} of the linear program~\eqref{EQ:OT}, we also describe a simple, parallelizable rounding algorithm that provably outputs a feasible solution to~\eqref{EQ:OT}. Specifically, for any $\eps>0$ and bounded, non-negative cost matrix $C$, we describe an algorithm that runs in time $\tildo(n^2/\eps^3)$ and outputs $\hat P \in \cU_{r,c}$ such that 
$$
 \langle \hat P, C \rangle \le \min_{P \in \cU_{r,c}} \langle  P, C \rangle + \eps
$$
We emphasize that our analysis does not require the cost matrix $C$ to come from an underlying metric; we only require $C$ to be non-negative. This implies that our results also give, for example, near-linear time approximation algorithms for Wasserstein $p$-distances between discrete measures.

{\bf Notation.}
We denote non-negative real numbers by $\re_+$, the set of integers $\{1, \dots, n\}$ by $[n]$, and the $n$-dimensional simplex by $\Delta_n := \{x \in \re_+^n \; : \; \sum_{i=1}^n x_i = 1 \}$. For two probability distributions $p,q \in \Delta_n$ such that $p$ is absolutely continuous w.r.t. $q$, we define the entropy $H(p)$ of $p$ and the Kullback-Leibler divergence $\KL(p \| q)$ between $p$ and $q$ respectively by
$$
H(p)=\sum_{i=1}^n p_i \log\left( \frac{1}{p_i}\right)\,, \qquad \KL(p \| q) := \sum_{i=1}^n p_i \log \left(\frac{p_i}{q_i}\right)\,.
$$ 
Similarly, for a matrix $P \in \re_+^{n \times n}$, we define the entropy $H(P)$ entrywise as $\sum_{ij}P_{ij} \log \frac{1}{P_{ij}}$.
We use $\bone$ and $\0$ to denote the all-ones and all-zeroes vectors in $\re^n$. For a matrix $A=(A_{ij})$, we denote by $\exp(A)$ the matrix with entries $(e^{A_{ij}})$. For $  A \in \re^{n \times n}$, we denote its row and columns sums by $r(A) := A\bone \in \R^n$ and $c(A) := A^\top\bone \in \R^n$, respectively. The coordinates $r_i(A)$ and $c_j(A)$ denote the $i$th row sum and $j$th column sum of $A$, respectively. We write $\|A\|_\infty=\max_{ij}|A_{ij}|$ and $\|A\|_1 = \sum_{ij} |A_{ij}|$.  For two matrices of the same dimension, we denote the Frobenius inner product of $A$ and $B$ by $\langle A, B \rangle = \sum_{ij} A_{ij} B_{ij}$.
For a vector $x \in \re^n$, we write $\diag(x) \in \re^{n \times n}$ to denote the diagonal matrix with  entries $(\diag(x))_{ii} = x_i$.
For any two nonnegative sequences $(u_n)_n, (v_n)_n$, we write $u_n=\tildo(v_n)$ if there exist positive constants $C,c$ such that $u_n\le Cv_n (\log n)^c$. For any two real numbers, we write $a\wedge b=\min(a,b)$.

\section{Optimal Transport in near-linear time}\label{sec:ot}

In this section, we describe the main algorithm studied in this paper. Pseudocode appears in Algorithm~\ref{alg:ot}.

\begin{wrapfigure}[18]{R}{0.48\textwidth} 
\begin{minipage}[t]{\linewidth}
\begin{algorithm}[H]
\caption{$\ot$($C$, $r$, $c$, $\eps$)}
\label{alg:ot}
\begin{algorithmic}[1]
\STATEx $\eta \leftarrow \frac{4 \log n}{\eps}$, $\eps' \leftarrow \frac{\eps}{8 \Cinfty}$
\STATEx \textbackslash\textbackslash$\;$ Step 1: Approximately project onto $\U$ 
\STATE $A \leftarrow \exp(-\eta C)$
\STATE $B \leftarrow \approxproj(A, \U, \eps')$
\STATEx
\STATEx \textbackslash\textbackslash$\;$Step 2: Round to feasible point in $\cU_{r,c}$
\STATE Output $\hat P \leftarrow \roundtrans(B, \U)$
\end{algorithmic}
\end{algorithm}
\vspace{-.25cm}
\end{minipage}
\begin{minipage}[t]{\linewidth}
\begin{algorithm}[H]
\caption{$\roundtrans(F, \U)$}
 \label{alg:round}
\begin{algorithmic}[1]
\STATE $X \leftarrow \diag(x)$ with $x_{i} = \frac{r_i}{r_i(F)}\wedge  1$
\STATE $F' \leftarrow XF$
\STATE $Y \leftarrow \diag(y)$ with $y_{j} = \frac{c_j}{c_j(F')}\wedge  1$
\STATE $F'' \leftarrow F'Y$
\STATE $\errr \leftarrow r - r(F'')$, $\errc \leftarrow c - c(F'')$
\STATE Output $G \leftarrow F'' + \errr \errc^\top/\|\errr\|_1$
\end{algorithmic}
\end{algorithm}
\end{minipage}

\end{wrapfigure}

The core of our algorithm is the computation of an \emph{approximate Sinkhorn projection} of the matrix $A= \exp(-\eta C)$ (Step~1), details for which will be given in Section~\ref{sec:sinkhorn}. %
Since our approximate Sinkhorn projection is not guaranteed to lie in the feasible set, we round our approximation to ensure that it lies in $\U$ (Step~2). Pseudocode for a simple, parallelizable rounding procedure is given in Algorithm~\ref{alg:round}.

Algorithm~\ref{alg:ot} hinges on two subroutines: $\approxproj$ and $\roundtrans$.
We give two algorithms for \approxproj: $\sinkhorn$\ and $\greedy$.
We devote Section~\ref{sec:sinkhorn} to their analysis, which is of independent interest.
On the other hand, $\roundtrans$ is fairly simple. Its analysis is postponed to Section~\ref{sec:proof-thm}. 

Our main theorem about Algorithm~\ref{alg:ot} is the following accuracy and runtime guarantee. The proof is postponed to Section~\ref{sec:proof-thm}, since it relies on the analysis of \approxproj\ and \roundtrans.
\begin{thm}\label{thm:ot}
Algorithm~\ref{alg:ot} returns a point $\hat P \in \U$ satisfying
\begin{equation*}
\langle \hat P, C \rangle \leq \min_{P \in \U} \langle P, C \rangle + \eps
\end{equation*}
in time $O(n^2 + S)$, where $S$ is the running time of the subroutine $\approxproj(A, \U, \eps')$.
In particular, if $\|C\|_\infty \leq L$, then $S$ can be $O(n^2L^3 (\log n)\eps^{-3})$, so that Algorithm~\ref{alg:ot} runs in  $O(n^2 L^3 (\log n)\eps^{-3})$ time.
\end{thm}

\begin{remark}
The time complexity in the above theorem reflects only elementary arithmetic operations. In the interest of clarity, we ignore questions of bit complexity that may arise from taking exponentials. The effect of this simplification is marginal since it can be easily shown~\cite{KalLarRic08} that the maximum bit complexity throughout the iterations of our algorithm is $O(L(\log n)/\eps)$. As a result, factoring in bit complexity leads to a runtime of $O(n^2 L^4 (\log n)^2\eps^{-4})$, which is still truly near-linear.
\end{remark}

\section{Linear-time approximate Sinkhorn projection}\label{sec:sinkhorn}
The core of our OT algorithm is the entropic penalty proposed by Cuturi~\cite{Cut13}:
\begin{equation}\label{eqn:peta}
\Peta := \argmin_{P \in \U} \big\{\langle P, C \rangle -  \eta^{-1} H(P)\big\}\,.
\end{equation}
The solution to~\eqref{eqn:peta} can be characterized explicitly by analyzing its first-order conditions for optimality.
\begin{lemma}{{\cite{Cut13}}}\label{lem:penalized_is_sinkhorn}
For any cost matrix $C$ and $r, c \in \Delta_n$, the minimization program~\eqref{eqn:peta}
has a unique minimum at $\Peta \in \U$ of the form $\Peta = X A Y$, where $A = \exp(-\eta C)$ and  $X, Y \in \R_+^{n \times n}$ are both  diagonal matrices. The matrices $(X,Y)$ are unique up to a constant factor.
\end{lemma}

We call the matrix $\Peta$ appearing in Lemma~\ref{lem:penalized_is_sinkhorn} the \emph{Sinkhorn projection} of $A$, denoted $\sink(A, \U)$, after Sinkhorn, who proved uniqueness in~\cite{sinkhorn}.
Computing $\sink(A, \U)$ exactly is impractical, so we implement instead an approximate version $\approxproj(A, \U, \eps')$, which outputs a matrix $B=X A Y$ that may not lie in $\U$ but satisfies the condition $\|r(B) - r\|_1 + \|c(B) - c\|_1 \leq \eps'$.
We stress that this condition is very natural from a statistical standpoint, since it requires that $r(B)$ and $c(B)$ are close to the target marginals $r$ and $c$ in \emph{total variation distance}.

\subsection{The classical Sinkhorn algorithm}\label{subsec:review-sinkhorn}
Given a matrix $A$, Sinkhorn proposed a simple iterative algorithm to approximate the Sinkhorn projection $\sink(A, \U)$, which is now known as the Sinkhorn-Knopp algorithm or RAS method.
Despite the simplicity of this algorithm and its good performance in practice, it has been difficult to analyze.
As a result, recent work showing that $\sink(A, \U)$ can be approximated in near-linear time~\cite{sinkhorn-wigderson,CohMadTsi17} has bypassed the Sinkhorn-Knopp algorithm entirely.\footnote{Replacing the $\approxproj$ step in Algorithm~\ref{alg:ot} with the matrix-scaling algorithm developed in~\cite{CohMadTsi17} results in a runtime that is a single factor of $\eps$ faster than what we present in Theorem~\ref{thm:ot}. The benefit of our approach is that it is extremely easy to implement, whereas the matrix-scaling algorithm of~\cite{CohMadTsi17} relies heavily on near-linear time Laplacian solver subroutines, which are not implementable in practice.}
In our work, we obtain a new analysis of the simple and practical Sinkhorn-Knopp algorithm, showing that it also approximates $\sink(A, \U)$ in near-linear time.

\begin{wrapfigure}[14]{R}{0.48\textwidth}
\begin{minipage}[t]{\linewidth}
\vspace{-.4cm}
\begin{algorithm}[H]
\caption{$\sinkhorn(A, \U, \eps')$} \label{alg:sinkhorn}
\begin{algorithmic}[1]
\STATE Initialize $k \leftarrow 0$
\STATE $A^{(0)} \leftarrow A/\|A\|_1$, $x^0 \leftarrow \0$, $y^0 \leftarrow \0$
\WHILE{$\dist(A^{(k)},\cU_{r, c}) > \eps'$}
\STATE $k \leftarrow k+ 1$
\IF{$k$ odd} 
\STATE $x_i \leftarrow \log \frac{r_i}{r_i(A^{(k-1)})}$ for $i \in [n]$
\STATE $x^k \leftarrow x^{k-1} + x$, $y^k \leftarrow y^{k-1}$
\ELSE \STATE $y \leftarrow \log \frac{c_j}{c_j(A^{(k-1)})}$ for $j \in [n]$
\STATE $y^{k} \leftarrow y^{k-1} + y$, $x^{k} \leftarrow x^{k-1}$
\ENDIF
\STATE $A^{(k)} = \diag(\exp(x^{k})) A \diag(\exp(y^{k}))$
\ENDWHILE
\STATE Output $B \leftarrow A^{(k)}$
\end{algorithmic}
\end{algorithm}
\end{minipage}
\end{wrapfigure}

Pseudocode for the Sinkhorn-Knopp algorithm appears in Algorithm~\ref{alg:sinkhorn}.
In brief, it is an alternating projection procedure which renormalizes the rows and columns of $A$ in turn so that they match the desired row and column marginals $r$ and $c$.
At each step, it prescribes to either modify all the rows by multiplying row $i$ by $r_i/r_i(A)$ for $i \in [n]$, or to do the analogous operation on the columns.
(We interpret the quantity $0/0$ as $1$ in this algorithm if ever it occurs.)
The algorithm terminates when the matrix $A^{(k)}$ is sufficiently close to the polytope $\U$.

\subsection{Prior work}
Before this work, the best analysis of Algorithm~\ref{alg:sinkhorn} showed that $\tildo((\eps')^{-2})$ iterations suffice to obtain a matrix close to $\U$ in $\ell_2$ distance:
\begin{prop}{{\cite{KalLarRic08}}}\label{prop:sinkhorn_old}
Let $A$ be a strictly positive matrix. Algorithm~\ref{alg:sinkhorn} with $\dist(A,\cU_{r, c}) = \|r(A) - r\|_2 + \|c(A) - c\|_2$ outputs a matrix $B$ satisfying $\|r(B) - r\|_2 + \|c(B) - c\|_2 \leq \eps'$ in $O\big(\rho(\eps')^{-2} \log (s/\ell)\big)$ iterations, where $s = \sum_{ij} A_{ij}$, $\ell = \min_{ij} A_{ij}$, and  $\rho > 0$ is such that $r_i, c_i \leq \rho$ for all $i \in [n]$.
\end{prop}

Unfortunately, this analysis is not strong enough to obtain a true near-linear time guarantee.
Indeed, the $\ell_2$ norm is not an appropriate measure of closeness between probability vectors, since very different distributions on large alphabets can nevertheless have small $\ell_2$ distance: for example, $(n^{-1}, \dots, n^{-1}, 0, \dots, 0)$ and $( 0, \dots, 0, n^{-1}, \dots, n^{-1})$ in $\Delta_{2n}$ have $\ell_2$ distance $\sqrt{2/n}$ even though they have disjoint support.
As noted above, for statistical problems, including computation of the OT distance, it is  more natural to measure distance in $\ell_1$ norm.

The following Corollary gives the best $\ell_1$ guarantee available from Proposition~\ref{prop:sinkhorn_old}.
\begin{corol}\label{cor:sinkhorn_old_l1}
Algorithm~\ref{alg:sinkhorn} with $\dist(A,\cU_{r, c}) = \|r(A) - r\|_2 + \|c(A) - c\|_2$ outputs a matrix $B$ satisfying $\|r(B) - r\|_1 + \|c(B) - c\|_1 \leq \eps'$ in $O\big(n\rho(\eps')^{-2} \log (s/\ell)\big)$ iterations.
\end{corol}
The extra factor of $n$ in the runtime of Corollary~\ref{cor:sinkhorn_old_l1} is the price to pay to convert an $\ell_2$ bound to an $\ell_1$ bound.
Note that $\rho \geq 1/n$, so $n\rho$ is always larger than $1$.
If $r=c=\bone_n/n$ are uniform distributions, then $n\rho=1$ and no dependence on the dimension appears.
However, in the extreme where $r$ or $c$ contains an entry of constant size, we get $n \rho = \Omega(n)$.

\subsection{New analysis of the Sinkhorn algorithm}\label{sec:analysis-sinkhorn}
Our new analysis allows us to obtain a dimension-independent bound on the number of iterations beyond the uniform case.
\begin{thm}\label{thm:sinkhorn}
Algorithm~\ref{alg:sinkhorn} with $\dist(A, \U) = \|r(A) - r\|_1 + \|c(A) - c\|_1$ outputs a matrix $B$ satisfying $\|r(B) - r\|_1 + \|c(B) - c\|_1 \leq \eps'$
in $O\big((\eps')^{-2} \log (s/\ell)\big)$
iterations, where $s = \sum_{ij} A_{ij}$ and $\ell = \min_{ij} A_{ij}$.
\end{thm}
Comparing our result with Corollary~\ref{cor:sinkhorn_old_l1}, we see what our bound is always stronger, by up to a factor of $n$.
Moreover, our analysis is extremely short.
Our improved results and simplified proof follow directly from the fact that we carry out the analysis entirely with respect to the Kullback-Leibler divergence, a common measure of statistical distance.
This measure possesses a close connection to the total-variation distance via Pinsker's inequality (Lemma~\ref{lem:pinsker}, below), from which we obtain the desired $\ell_1$ bound. Similar ideas can be traced back at least to~\cite{GurYia98} where an analysis of Sinkhorn iterations for bistochastic targets is sketched in the context of a different problem: detecting the existence of a perfect matching in a bipartite graph.

We first define some notation.
Given a matrix $A$ and desired row and column sums $r$ and $c$, we define the potential (Lyapunov) function $f: \re^n \times \re^n \to \re$ by
\begin{equation*}
f(x, y) = \sum_{ij} A_{ij} e^{x_i + y_j} - \langle r, x \rangle - \langle c, y \rangle\,.
\end{equation*}
This auxiliary function has appeared in much of the literature on Sinkhorn projections~\cite{KalLarRic08,CohMadTsi17,KalKha96,KalKha93}.
We call the vectors $x$ and $y$ \emph{scaling vectors}. It is easy to check that a minimizer $(x^*, y^*)$ of $f$ yields the Sinkhorn projection of $A$: writing $X = \diag(\exp(x^*))$ and $Y = \diag(\exp(y^*))$, first order optimality conditions imply that $X A Y$ lies in $\U$, and therefore $X A Y = \sink(A, \U)$.
\par The following lemma exactly characterizes the improvement in the potential function $f$ from an iteration of Sinkhorn, in terms of our current divergence to the target marginals.

\begin{lemma}\label{lem:improvement}
If $k \geq 2$, then $\DS
f(x^{k-1}, y^{k-1}) - f(x^{k}, y^{k}) = \KL(r \| r(A^{(k-1)})) + \KL(c \| c(A^{(k-1)}))\,.$
\end{lemma}
\begin{proof}
Assume without loss of generality that $k$ is odd, so that $c(A^{(k-1)}) = c$ and $r(A^{(k)}) = r$. (If $k$ is even, interchange the roles of $r$ and $c$.) By definition,
\begin{align*}
f(x^{k-1}, y^{k-1}) - f(x^{k}, y^{k}) & = \sum_{ij}\big( A^{(k-1)}_{ij} - A^{(k)}_{ij} \big)+ \langle r, x^{k} - x^{k-1} \rangle + \langle c, y^{k} - y^{k-1} \rangle \\
& = \sum_{i} r_i (x^k_i - x^{k-1}_i) = \KL(r \|r(A^{(k-1)}) + \KL(c \| c(A^{(k-1)})\,,
\end{align*}
where we have used that: $\|A^{(k-1)}\|_1 = \|A^{(k)}\|_1 = 1$ and $Y^{(k)} = Y^{(k-1)}$; for all $i$, $r_i(x^k_i - x^{k-1}_i) = r_i \log \frac{r_i}{r_i(A^{(k-1)})}$; and $\KL(c \| c(A^{(k-1)})) = 0$ since $c = c(A^{(k-1)})$.
\end{proof}
The next lemma has already appeared in the literature and we defer its proof to the Appendix.
\begin{lemma}\label{lem:starting}
If $A$ is a positive matrix with $\|A\|_1 \leq s$ and smallest entry $\ell$, then
\begin{equation*}
f(x^1, y^1) - \min_{x, y \in \re} f(x, y)  \leq f(0,0) - \min_{x, y \in \re} f(x, y) \leq \log \frac{s}{\ell}\,.
\end{equation*}
\end{lemma}
\begin{lemma}[Pinsker's Inequality]\label{lem:pinsker}
For any probability measures $p$ and $q$, $\|p - q\|_1 \leq \sqrt{2 \KL(p \| q)}$.
\end{lemma}

\begin{proof}[Proof of Theorem~\ref{thm:sinkhorn}]
Let $k^*$ be the first iteration such that $\|r(A^{(k^*)}) - r\|_1 + \|c(A^{(k^*)}) - c\|_1 \leq \eps'$.
Pinsker's inequality implies that for any $k<k^*$, we have
\begin{equation*}
\eps'^2 < (\|r(A^{(k)}) - r\|_1 + \|c(A^{(k)}) - c\|_1)^2 \leq 4 (\KL(r \| r(A^{(k)}) + \KL(c \| c(A^{(k)}))\,,
\end{equation*}
so Lemmas~\ref{lem:improvement} and~\ref{lem:starting} imply that we terminate in $k^*\le 4\eps'^{-2} \log (s/\ell)$ steps, as claimed.
\end{proof}

\subsection{Greedy Sinkhorn}\label{subsec:greedy}
In addition to a new analysis of \sinkhorn, we propose a new algorithm \greedy{} which enjoys the same convergence guarantee but performs better in practice. 
Instead of performing alternating updates of \textit{all} rows and columns of $A$, the \greedy{} algorithm updates only a \textit{single} row or column at each step.
Thus \greedy{} updates only $O(n)$ entries of $A$ per iteration, rather than $O(n^2)$.

In this respect, \greedy{} is similar to the stochastic algorithm for Sinkhorn projection proposed by~\cite{GenCutPey16}.
There is a natural interpretation of both algorithms as coordinate descent algorithms in the dual space corresponding to row/column violations.
Nevertheless, our algorithm differs from theirs in several key ways. Instead of choosing a row or column to update randomly, \greedy{} chooses the best row or column to update greedily. Additionally, \greedy{} does an exact line search on the coordinate in question since there is a simple closed form for the optimum, whereas the algorithm proposed by~\cite{GenCutPey16} updates in the direction of the average gradient. Our experiments establish that \greedy{} performs better in practice; more details appear in the Appendix.

We emphasize that our algorithm is an extremely natural modification of \sinkhorn{}, and greedy algorithms for the scaling problem have been proposed before, though these do not come with with explicit near-linear time guarantees~\cite{ParLan82}. However, whereas previous analyses of \sinkhorn{} cannot be modified to extract any meaningful rates of convergence for greedy algorithms, our new analysis of \sinkhorn{} from Section~\ref{sec:analysis-sinkhorn} applies to \greedy{} with only trivial modifications.

\begin{wrapfigure}[13]{R}{0.5\textwidth}
\begin{minipage}[t]{\linewidth}
\vspace{-.8cm}
\begin{algorithm}[H]
\caption{$\greedy(A, \U, \eps')$} \label{alg:greedy}
\begin{algorithmic}[1]
\STATE $A^{(0)} \leftarrow A/\|A\|_1$, $x \leftarrow \0$, $y \leftarrow \0$.
\STATE $A \leftarrow A^{(0)}$
\WHILE{$\dist(A,\cU_{r, c}) > \eps$}
\STATE $I \leftarrow \argmax_i \rho(r_i, r_i(A))$
\STATE $J \leftarrow \argmax_j \rho(c_j, c_j(A))$
\IF{$\rho(r_I, r_I(A)) > \rho(c_J, c_J(A))$}
\STATE $x_I \leftarrow x_I + \log \frac{r_I}{r_I(A)}$
\ELSE
\STATE $y_J \leftarrow y_J + \log \frac{c_J}{c_J(A)}$
\ENDIF
\STATE $A \leftarrow \diag(\exp(x))A^{(0)}\diag(\exp(y))$
\ENDWHILE
\STATE Output $B \leftarrow A$
\end{algorithmic}
\end{algorithm}
\end{minipage}
\end{wrapfigure}

Pseudocode for \greedy{} appears in Algorithm~\ref{alg:greedy}. We define $\dist(A,\cU_{r, c}) = \|r(A) - r\|_1 + \|c(A) - c\|_1$ and define the distance function $\rho: \re_+ \times \re_+ \to [0, +\infty]$ by
\begin{equation*}
\rho(a, b) = b - a + a \log \frac a b\,.
\end{equation*}
The choice of $\rho$ is justified by its appearance in Lemma~\ref{lem:greedy_improvement}, below.
While $\rho$ is not a metric, it is easy to see that $\rho$ is nonnegative and satisfies $\rho(a, b) = 0$ iff $a = b$.

We note that after $r(A)$ and $c(A)$ are computed once at the beginning of the algorithm, \greedy{} can easily be implemented such that each iteration runs in only $O(n)$ time.

\begin{thm}\label{thm:greedy}
The algorithm \greedy{} outputs a matrix $B$ satisfying $\|r(B) - r\|_1 + \|c(B) - c\|_1 \leq \eps'$ in
$
O(n(\eps')^{-2}\log (s/\ell))$
iterations, where  $s = \sum_{ij} A_{ij}$ and $\ell = \min_{ij} A_{ij}$. Since each iteration takes $O(n)$ time, such a matrix can be found in $O(n^2(\eps')^{-2} \log (s/\ell))$ time.
\end{thm}

The analysis requires the following lemma, which is an easy modification of Lemma~\ref{lem:improvement}.
\begin{lemma}\label{lem:greedy_improvement}
Let $A'$ and $A''$ be successive iterates of \greedy{}, with corresponding scaling vectors $(x', y')$ and $(x'', y'')$.
If $A''$ was obtained from $A'$ by updating row $I$, then
\begin{equation*}
f(x', y') - f(x'', y'') = \rho(r_I, r_I(A'))\,,
\end{equation*}
and if it was obtained by updating column $J$, then
\begin{equation*}
f(x', y') - f(x'', y'') = \rho(c_J, c_J(A'))\,.
\end{equation*}
\end{lemma}
We also require the following extension of Pinsker's inequality (proof in Appendix).
\begin{lemma}\label{lem:pinsker_extension}
For any  $\alpha \in \Delta_n, \beta \in \R_+^n$, define $\rho(\alpha, \beta)=\sum_i \rho(\alpha_i, \beta_i)$.
If  $\rho(\alpha, \beta) \leq 1$, then
\begin{equation*}
\|\alpha - \beta\|_1 \leq \sqrt{7 \rho(\alpha, \beta)}\,.
\end{equation*}
\end{lemma}
\begin{proof}[Proof of Theorem~\ref{thm:greedy}]
We follow the proof of Theorem~\ref{thm:sinkhorn}. 
Since the row or column update is chosen greedily, at each step we make progress of at least
$\frac{1}{2n} (\rho(r, r(A)) + \rho(c, c(A)))$.
If $\rho(r, r(A))$ and $\rho(c, c(A))$ are both at most $1$, then under the assumption that $\|r(A) - r\|_1 + \|c(A) - c\|_1 > \eps'$, our progress is at least
\begin{equation*}
\frac{1}{2n} (\rho(r, r(A)) + \rho(c, c(A))) \geq \frac{1}{14 n} (\|r(A) - r\|_1^2 + \|c(A) - c\|_1^2) \geq \frac{1}{28 n} \eps'^2
\end{equation*}
Likewise, if either $\rho(r, r(A))$ or $\rho(c, c(A))$ is larger than $1$, our progress is at least $1/2n \geq \frac{1}{28 n} \eps'^2$.
Therefore, we terminate in at most $28 n \eps'^{-2} \log (s/\ell)$ iterations.
\end{proof}

\section{Proof of Theorem~\ref{thm:ot}}\label{sec:proof-thm}
First, we present a simple guarantee about the rounding Algorithm~\ref{alg:round}. The following lemma shows that
the $\ell_1$ distance between the input matrix $F$ and rounded matrix $G = \roundtrans(F, \U)$ is controlled by the total-variation distance between the input matrix's marginals $r(F)$ and $c(F)$ and the desired marginals $r$ and $c$.
\begin{lemma}\label{lem:roundtrans}
If $r,c \in \Delta_n$ and $F \in \re_+^{n \times n}$, then Algorithm~\ref{alg:round} takes $O(n^2)$ time to output a matrix $G \in \U$ satisfying
\[
\|G - F\|_1 \leq
2 \Big[ \left\|r(F) -r \right\|_1 +
 \left\|c(F) -c \right\|_1 \Big]
\,.
\]
\end{lemma}
The proof of Lemma~\ref{lem:roundtrans} is simple and left to the Appendix. (We also describe in the Appendix a randomized variant of Algorithm~\ref{alg:round} that achieves a slightly better bound than Lemma~\ref{lem:roundtrans}). We are now ready to prove Theorem~\ref{thm:ot}.
\begin{proof}[Proof of Theorem~\ref{thm:ot}]
\textsc{Error analysis.} Let $B$ be the output of \linebreak$\approxproj(A, \U, \eps')$, and let $P^* \in \argmin_{P \in \U} \langle P, C \rangle$ be an optimal solution to the original OT program.

We first show that $\langle B, C \rangle$ is not much larger than $\langle P^*, C \rangle$.
To that end, write $r' := r(B)$ and $c' := c(B)$. Since $B = X A Y$ for positive diagonal matrices $X$ and $Y$, Lemma~\ref{lem:penalized_is_sinkhorn} implies $B$ is the optimal solution to
\begin{equation}\label{eqn:prime_program}
\min_{P \in \Uprime} \langle P, C\rangle - \eta^{-1} H(P)\,.
\end{equation}
By Lemma~\ref{lem:roundtrans}, there exists a matrix $P' \in \Uprime$ such that
$$\|P' - P^*\|_1 \leq 2\left( \|r' - r\|_1 + \|c' - c\|_1\right)\,.$$ Moreover,
since $B$ is an optimal solution of~\eqref{eqn:prime_program}, we have
\begin{equation*}
\langle B, C\rangle - \eta^{-1} H(B) \leq \langle P', C\rangle - \eta^{-1} H(P')\,.
\end{equation*}
Thus, by H\"older's inequality
\begin{align}
\nonumber \langle B, C \rangle - \langle P^*, C \rangle & = \langle B, C \rangle - \langle P', C \rangle + \langle P', C \rangle  - \langle P^*, C \rangle \\
\nonumber & \leq  \eta^{-1} (H(B) - H(P')) + 2(\|r' - r\|_1 + \|c' - c\|_1)\Cinfty \\
& \leq {2\eta^{-1} \log n} + 2(\|r' - r\|_1 + \|c' - c\|_1)\Cinfty\,, \label{eqn:qpbound}
\end{align}
where we have used the fact that $0\le H(B), H(P')\le  2 \log n$.

Lemma~\ref{lem:roundtrans} implies that the output $\hat P$ of $\roundtrans(B, \U)$ satisfies the inequality $
\|B - \hat P\|_1 \leq 2\left(\|r' - r\|_1 + \|c' - c\|_1\right)$.
This fact together with~\eqref{eqn:qpbound} and  H\"older's inequality yields
\begin{equation*}
\langle \hat P, C \rangle \leq \min_{P \in \U} \langle P, C \rangle + {2 \eta^{-1}\log n} + 4 (\|r' - r\|_1 + \|c' - c\|_1)\Cinfty\,.
\end{equation*}

Applying the guarantee of $\approxproj(A, \U, \eps')$, we obtain
\begin{equation*}
\langle \hat P, C \rangle \leq \min_{P \in \U} \langle P, C \rangle + \frac{2 \log n}{\eta} + 4 \eps'\Cinfty\,.
\end{equation*}

Plugging in the values of $\eta$ and $\eps'$ prescribed in Algorithm~\ref{alg:ot} finishes the error analysis.

\textsc{Runtime analysis.}
Lemma~\ref{lem:roundtrans} shows that Step 2 of Algorithm~\ref{alg:ot} takes $O(n^2)$ time. The runtime of Step 1 is dominated by the $\approxproj(A, \U, \eps')$ subroutine. Theorems~\ref{thm:sinkhorn} and~\ref{thm:greedy} imply that both the \sinkhorn{} and \greedy{} algorithms accomplish this in $S = O(n^2(\eps')^{-2} \log \frac{s}{\ell})$ time, where $s$ is the sum of the entries of $A$ and $\ell$ is the smallest entry of $A$.
Since the matrix $C$ is nonnegative, the entries of $A$ are bounded above by $1$, thus $s \leq n^2$. The smallest entry of $A$ is $e^{-\eta \Cinfty}$, so $\log 1/\ell = \eta \Cinfty$.
We obtain $S = O(n^2 (\eps')^{-2} (\log n + \eta \Cinfty))$.
The proof is finished by plugging in the values of $\eta$ and $\eps'$ prescribed in Algorithm~\ref{alg:ot}.
\end{proof}

\section{Empirical results}\label{sec:empirical}

\begin{wrapfigure}[10]{R}{0.3\textwidth} 
\centering
\includegraphics[width=.2\textwidth]{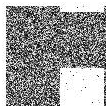}
\caption{Synthetic image.}
\label{fig:rand_images_view}
\end{wrapfigure}
Cuturi~\cite{Cut13} already gave experimental evidence that using $\sinkhorn$ to solve~\eqref{eqn:peta} outperforms state-of-the-art techniques for optimal transport. In this section, we provide strong empirical evidence that our proposed $\greedy$ algorithm significantly outperforms $\sinkhorn$.

We consider transportation between pairs of $m \times m$ grayscale images, normalized to have unit total mass. The target marginals $r$ and $c$ represent two images in a pair, and $C\in \R^{m^2\times m^2}$ is the matrix of $\ell_1$ distances between pixel locations. Therefore, we aim to compute the earth mover's distance.
\par We run experiments on two datasets: \textit{real images}, from \textsc{mnist}, and \textit{synthetic images}, as in Figure~\ref{fig:rand_images_view}.

\subsection{MNIST}\label{subsec:mnist}
We first compare the behavior of $\greedy$ and $\sinkhorn$ on real images. To that end, we choose $10$ random pairs of images from the MNIST dataset, and for each one analyze the performance of $\ot$ when using both $\greedy$ and $\sinkhorn$ for the approximate projection step. We add negligible noise $0.01$ to each background pixel with intensity $0$. Figure~\ref{fig:images_plots} paints a clear picture: $\greedy$ significantly outperforms $\sinkhorn$ both in the short and long term.

\begin{wrapfigure}[22]{R}{0.6\textwidth}
\vspace{-.5cm}
\centering
\includegraphics[width=0.5\linewidth]{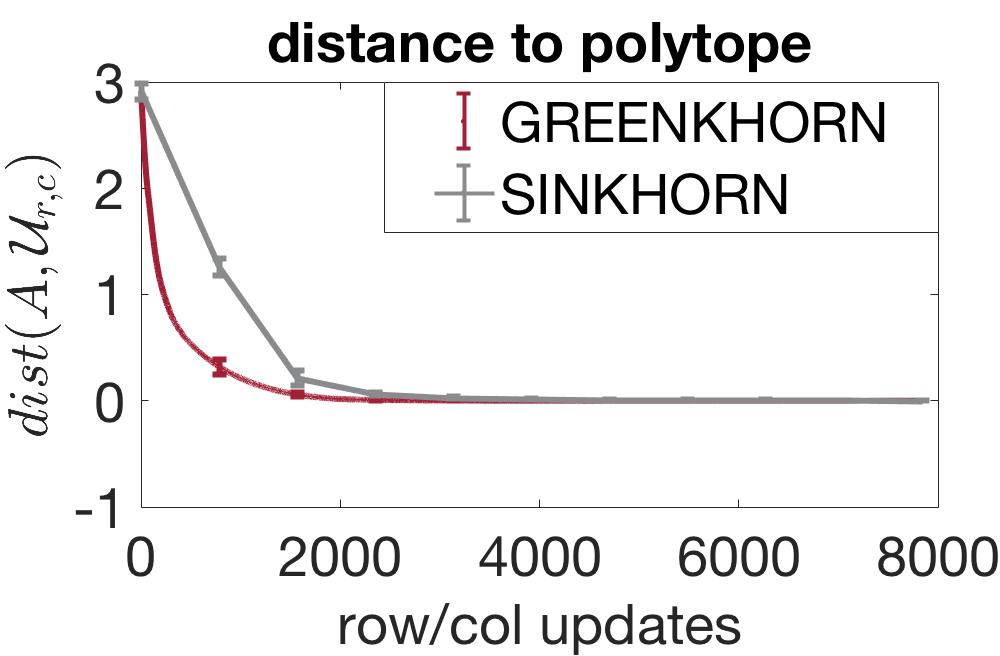}\hfill
\includegraphics[width=0.5\linewidth]{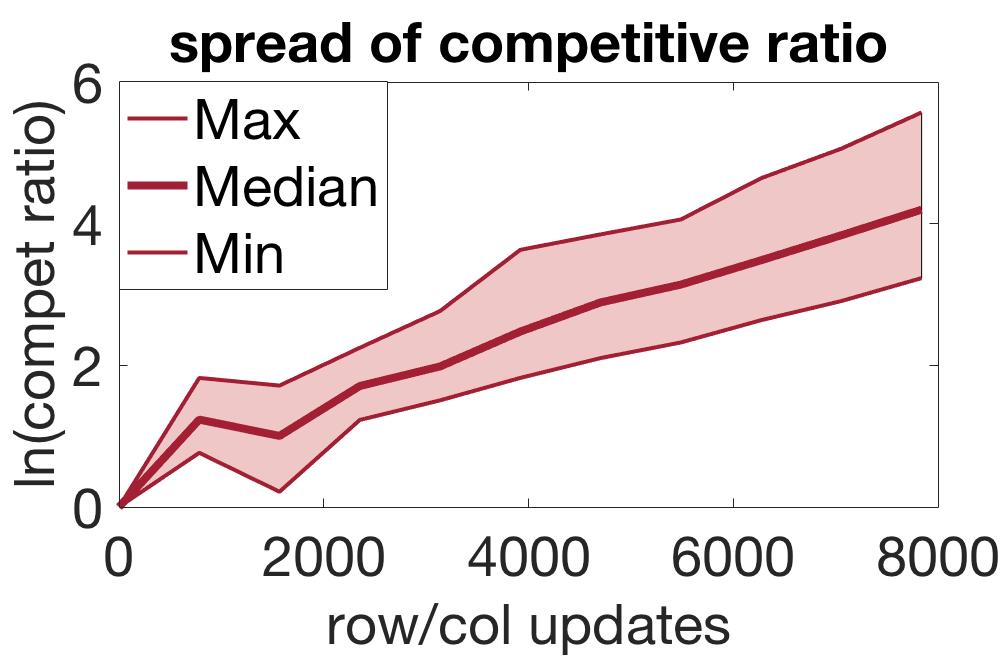}
\includegraphics[width=0.5\linewidth]{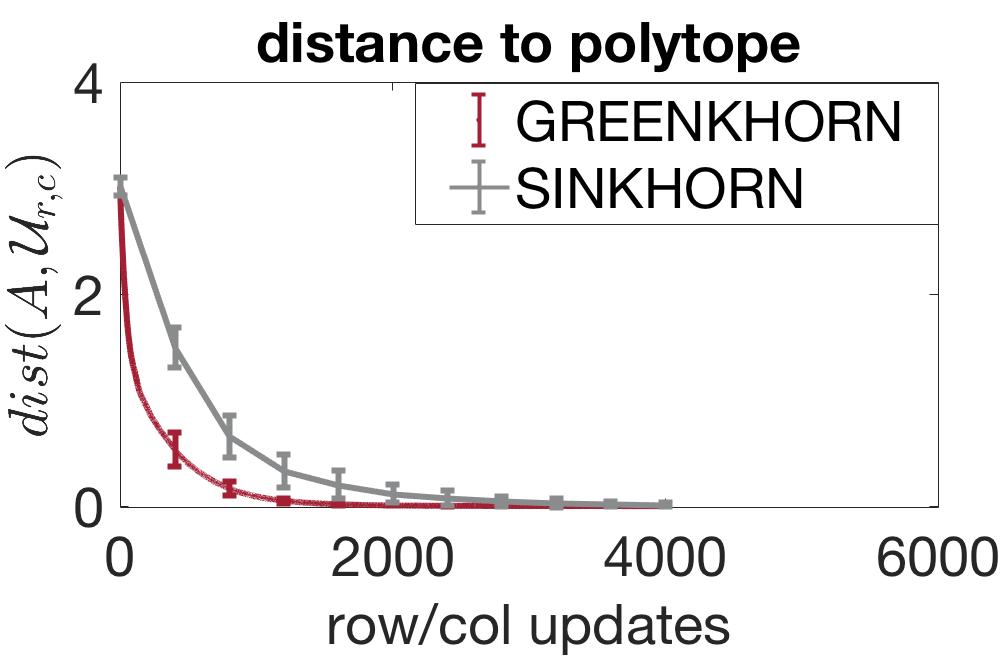}\hfill
\includegraphics[width=0.5\linewidth]{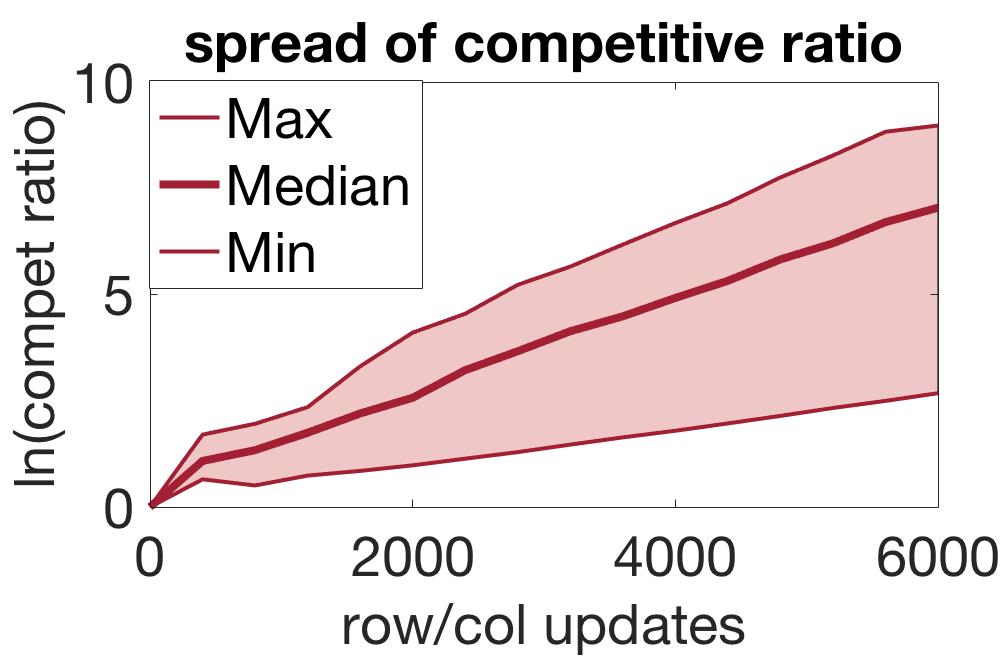}
\caption{Comparison of $\greedy$ and $\sinkhorn$ on pairs of MNIST images of dimension $28 \times 28$ (top) and random images of dimension $20 \times 20$ with $20$\% foreground (bottom). Left: distance $\dist(A, \cU_{r, c})$ to the transport polytope (average over $10$ random pairs of images). Right: maximum, median, and minimum values of the competitive ratio
 $\ln\left(\dist(A_{\sc S}, \U)/\dist(A_{\sc G}, \U)\right)$
 over $10$ runs.}
\label{fig:images_plots}
\end{wrapfigure}

\newpage
\subsection{Random images}\label{subsec:rand-images} To better understand the empirical behavior of both algorithms in a number of different regimes, we devised a synthetic and tunable framework whereby we generate images by choosing a randomly positioned ``foreground'' square in an otherwise black background. The size of this square is a tunable parameter varied between 20\%, 50\%, and 80\% of the total image's area. Intensities of background pixels are drawn uniformly from $[0,1]$; foreground pixels are  drawn uniformly from $[0,50]$. Such an image is depicted in Figure~\ref{fig:rand_images_view}, and results appear in Figure~\ref{fig:images_plots}.

We perform two other experiments with random images in Figure~\ref{fig:rand_experiments}. In the first, we vary the number of background pixels and show that $\greedy$ performs better when the number of background pixels is larger.
We conjecture that this is related to the fact that $\greedy$ only updates salient rows and columns at each step, whereas $\sinkhorn$ wastes time updating rows and columns corresponding to background pixels, which have negligible impact.
This demonstrates that $\greedy$ is a better choice especially when data is sparse, which is often the case in practice.

In the second, we consider the role of the regularization parameter $\eta$. Our analysis requires taking $\eta$ of order $\log n/\eps$, but Cuturi~\cite{Cut13} observed that in practice $\eta$ can be much smaller. Cuturi showed that \sinkhorn{} outperforms state-of-the art techniques for computing OT distance even when $\eta$ is a small constant, and Figure~\ref{fig:rand_experiments} shows that \greedy{} runs faster than \sinkhorn{} in this regime with no loss in accuracy.

\begin{figure}[ht]
\centering
\includegraphics[height=96pt]{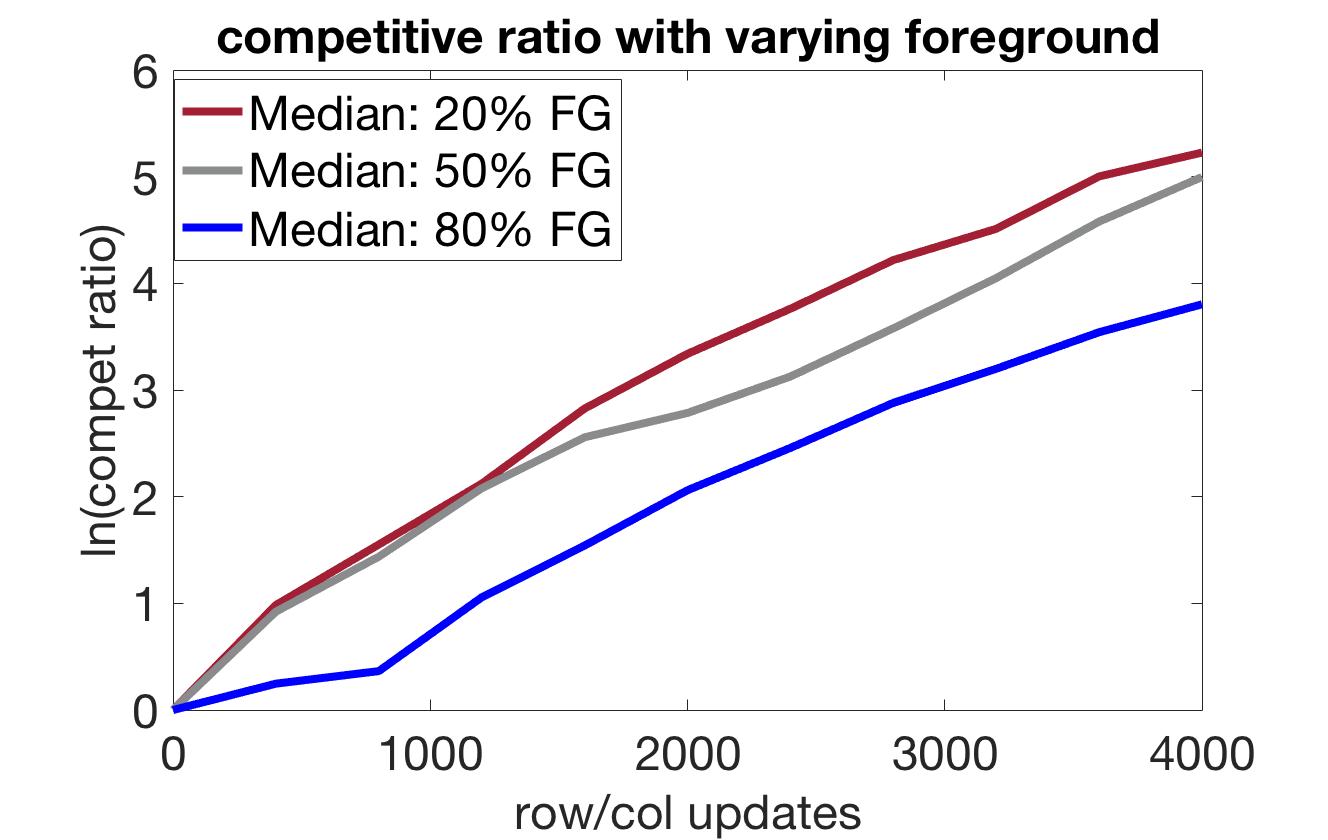}\hfill
\includegraphics[height=96pt]{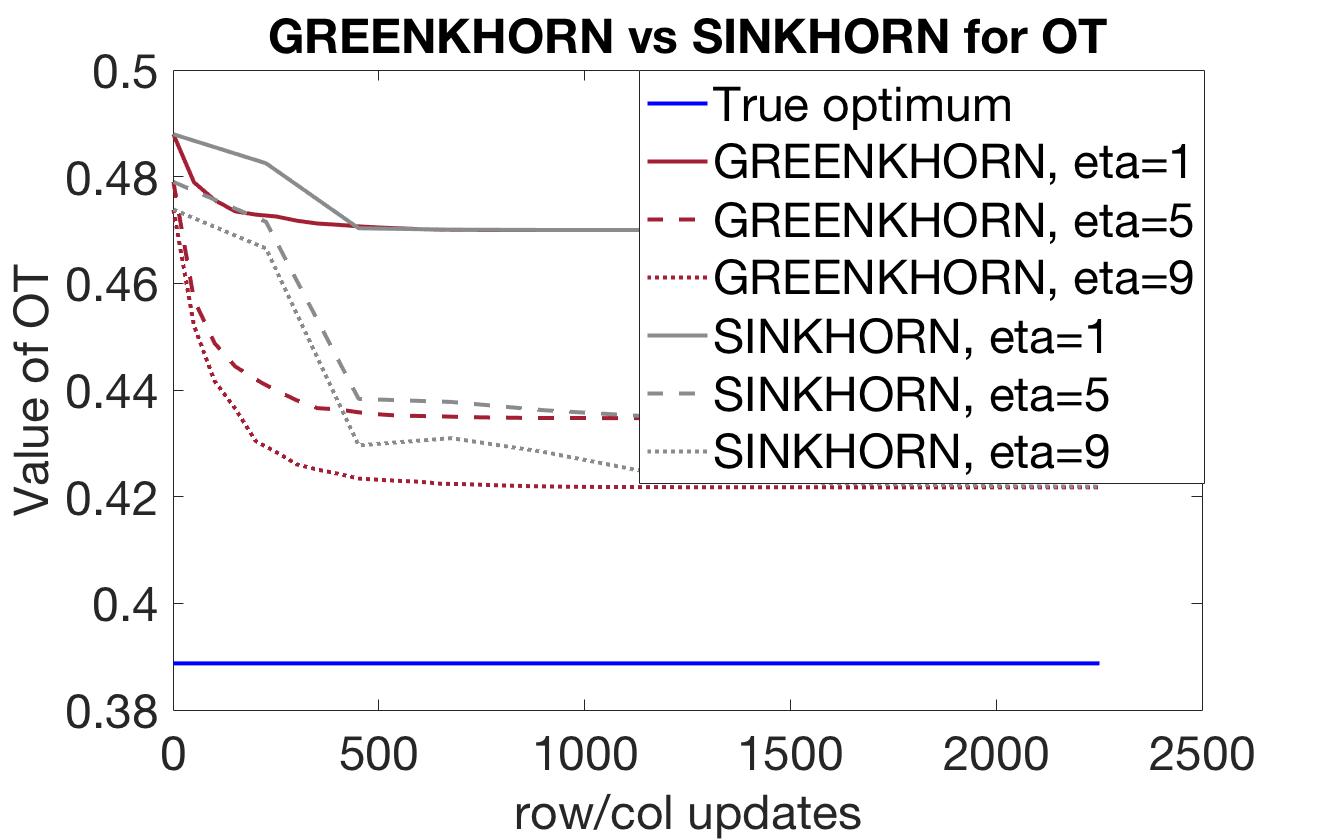}
\caption{Left: Comparison of median competitive ratio for random images containing $20$\%, $50$\%, and $80$\% foreground. Right: Performance of \greedy{} and \sinkhorn{} for small values of $\eta$.
}
\label{fig:rand_experiments}
\end{figure}

\appendix
\section{Omitted proofs}
\subsection{Proof of Lemma~\ref{lem:starting}}
The proof of the first inequality is similar to the proof of Lemma~\ref{lem:improvement}:
\begin{align*}
f(0, 0) - f(x^{(1)}, y^{(1)}) & = \langle r, x^{(1)} \rangle  + \langle c, y^{(1)} \rangle  = \sum_{ij} A^{(1)}_{ij} \log \frac{A^{(1)}_{ij}}{A^{(0)}_{ij}}  = \KL(A^{(1)} \| A^{(0)}) \geq 0\,,
\end{align*}
where $\KL(A^{(1)} \| A^{(0)})$ denotes the divergence between $A^{(1)}$ and $A^{(0)}$ viewed as elements of $\Delta_{n^2}$.

We now prove the second claim. Note that $A^{(0)}$ satisfies $\|A^{(0)}\|_1=1$ and has smallest entry $\ell/s$.
Since $A^{(0)}$ is positive, \cite{sinkhorn} shows that $\sink(A^{(0)})$ exists and is unique. Let $(x^*, y^*)$ be corresponding scaling factors. Then
\begin{align*}
f(0, 0) - f(x^*, y^*) & = \langle r, x^{*} \rangle  + \langle c, y^{*} \rangle\,.
\end{align*}
Now since
\begin{equation*}
A^{(0)}_{ij}e^{x^*_i + y^*_j} \leq \sum_{ij} A^{(0)}_{ij}e^{x^*_i + y^*_j} = 1\,,
\end{equation*}
we have
\begin{equation*}
x^*_i + y^*_j \leq \log \frac s \ell\,,
\end{equation*}
for all $i, j \in [n]$. Thus because $r$ and $c$ are both probability vectors,
\begin{equation*}
\langle r, x^* \rangle + \langle c, y^*\rangle \leq \log \frac s \ell\,.
\end{equation*}
\qed

\subsection{Proof of Lemma~\ref{lem:greedy_improvement}}
We prove only the case where a row was updated, since the column case is exactly the same.

By definition,
\begin{equation*}
f(x', y') - f(x'', y'') = \sum_{ij}( A'_{ij} - A''_{ij}) + \langle r, x'' - x' \rangle + \langle c, y'' - y' \rangle\,.
\end{equation*}
Observe that  $A'$ and $A''$ differ only in the $I$th row, and $x''$ and $x'$ differ only in the $I$th entry, and $y'' = y'$. Hence
\begin{align*}
f(x', y') - f(x'', y'') & = r_I(A') - r_I(A'') + r_I(x''_I - x'_I) \\
& = \rho(r_I, r_I(A'))\,,
\end{align*}
where we have used the fact that $r_I(A'') = r_I$ and $x''_I - x'_I = \log (r_I/r_I(A'))$.
\qed

\subsection{Proof of Lemma~\ref{lem:pinsker_extension}}
Let $s = \sum_i \beta_i$, and write $\bar \beta = \beta/s$.
The definition of $\rho$ implies
\begin{align*}
\rho(\alpha, \beta) & = \sum_i (\beta_i - \alpha_i) + \alpha_i \log \frac{\alpha_i}{\beta_i} \\
& = s - 1 + \sum_i \alpha_i \log \frac{\alpha_i}{s \bar \beta_i} \\
& = s - 1 -(\log s) \sum_i \alpha_i  + \KL(\alpha \| \bar \beta) \\
& = s - 1 - \log s + \KL(\alpha \| \bar \beta)\,.
\end{align*}
Note that both $s - 1 - \log s$ and $\KL(\alpha \| \bar \beta)$ are nonnegative.
If $\rho(\alpha, \beta) \leq 1$, then in particular $s - 1 - \log s \leq 1$, and it can be seen that  $s-1 - \log s \geq (s-1)^2/5$ in this range.
Applying Lemma~\ref{lem:pinsker} (Pinsker's inequality) yields
\begin{equation*}
\rho(\alpha, \beta) \geq \frac 1 5 (s-1)^2 + \frac 1 2 \|\alpha - \bar \beta\|_1^2\,.
\end{equation*}
By the triangle inequality and convexity,
\begin{equation*}
\|\alpha - \beta\|_1^2 \leq (\|\bar \beta - \beta\|_1 + \|\alpha - \bar \beta\|_1)^2 = (|s -1| + \|\alpha - \bar \beta\|_1)^2 \leq \frac 7 5 (s-1)^2 + \frac 7 2 \|\alpha - \bar \beta\|_1^2\,.
\end{equation*}
The claim follows from the above two displays.
\qed

\subsection{Proof of Lemma~\ref{lem:roundtrans}}
Let $G$ be the output of $\roundtrans(F, \U)$.
The entries of $F''$ are nonnegative, and at the end of the algorithm $\errr$ and $\errc$ are both nonnegative, with $\|\errr\|_1 = \|\errc\|_1 = 1 - \|F''\|_1$.
Therefore the entries of $G$ are nonnegative and
\[
r(G) = r(F'') + r(\errr \errc^\top/\|\errr\|_1) = r(F'') + \errr = r\,,
\]
and likewise $c(G) = c$. This establishes that $G \in \U$.

Now we prove the $\ell_1$ bound between the original matrix $F$ and $G$. Let $\Delta = \|F\|_1 - \|F''\|_1$ be the total amount of mass removed from $F$ by rescaling the rows and columns.
In the first step, we remove mass from a row of $F$ when $r_i(F) \geq r_i$, and in the second step we remove mass from a column when $c_j(F') \geq c_j$.
We therefore have
\begin{align}
\Delta &= \sum_{i=1}^n (r_i(F) - r_i)_+ + \sum_{j=1}^n (c_j(F') - c_j)_+\,. \label{eq:proof-lem-7:delta}
\end{align}
Let us analyze both of the sums in~\eqref{eq:proof-lem-7:delta}. First, a simple calculation shows
\begin{align*}
\sum_{i=1}^n (r_i(F) - r_i)_+ = \frac 1 2\Big[\|r(F)-r\|_1 + \|F\| - 1\Big]\,.
\end{align*}
Next, upper bound the second sum in~\eqref{eq:proof-lem-7:delta} using the fact that the vector $c(F)$ is entrywise larger than $c(F')$
\begin{align*}
\sum_{j=1}^n (c_j(F') - c_j)_+\
\leq
\sum_{j=1}^n (c_j(F) - c_j)_+\
\leq 
\|c(F) - c\|_1
\end{align*}

Therefore we conclude
\begin{align}
\|G-F\|_1 
& \leq \Delta + \|\errr \errc^\top\|_1/\|\errr\|_1 \nonumber \\
& = \Delta + 1 - \|F''\|_1 \nonumber \\
& = 2 \Delta + 1 - \|F\|_1 \nonumber \\
& \leq \|r(F)-r\|_1 + 2\|c(F) - c\|_1 \label{eq:proof-lem-7:bound} \\
& \leq 2 \Big[\|r(F)-r\|_1 + \|c(F) - c\|_1\Big] \nonumber
\end{align}

\par Finally, we prove the $O(n^2)$ runtime bound follows by observing that each rescaling  and computing the matrix $\errr \errc^\top/\|\errr\|_1$  both require at most $O(n^2)$ time.
\qed

\subsection{Randomized variant of rounding algorithm (Algorithm~\ref{alg:round})}
In the section, we describe a simple randomized variant of Algorithm~\ref{alg:round} that achieves a slightly better guarantee. Let us first recall the guarantee we get for Algorithm~\ref{alg:round}. By equation~\eqref{eq:proof-lem-7:bound} in the proof of Lemma~\ref{lem:roundtrans}, the $\ell_1$ difference between the original matrix $F$ and rounded matrix $G$ is upper bounded by
\begin{align*}
\|G - F\|_1 \leq \|r(F) - r\|_1 + 2\|c(F) - c\|_1\,.
\end{align*}
This asymmetry between $\|r(F) - r\|_1$ and $\|c(F) - c\|_1$ arises because Algorithm~\ref{alg:round} creates $F''$ by first removing mass from rows of $F$, and then from columns. Consider modifying Algorithm~\ref{alg:round} to create $F''$ by first removing mass from columns of $F$, and then from rows. Then a symmetrical argument gives the bound
\begin{align*}
\|G - F\|_1 \leq 2\|r(F) - r\|_1 + \|c(F) - c\|_1\,.
\end{align*}
Together the above two displays suggest the following simple randomized variant of Algorithm~\ref{alg:round}: with probability $1/2$, perform Algorithm~\ref{alg:round}; otherwise, perform the above-described column-then-row version of Algorithm~\ref{alg:round}. Combining the above two displays then gives the following improved bound for this randomized algorithm
\begin{align*}
\E\|G - F\|_1 \leq \frac{3}{2} \Big[\|r(F) - r\|_1 + \|c(F) - c\|_1 \Big]\,.
\end{align*}

\subsection{Comparison with~\cite{GenCutPey16}}
In this Section, we present an empirical comparison of the performance of \greedy{} with the stochastic algorithm proposed by~\cite{GenCutPey16}.
Their \linebreak algorithm---which we call Stochastic Sinkhorn for convenience---uses a Stochastic Averaged Gradient (SAG) algorithm to optimize a dual version of the entropic penalty program~\eqref{eqn:peta}.

We have noted in the main text that \greedy{} and Stochastic Sinkhorn both attempt to solve the scaling problem via coordinate descent in the dual problem. Stochastic Sinkhorn does so via the method proposed in~\cite{SchLerBac17}, whereas \greedy{} greedily chooses a good coordinate to update, and then leverages an explicit closed form to perform an exact line search on this coordinate. One difference between our algorithms is their starting point: \greedy{} is initialized with $A/\|A\|_1$, whereas the starting primal solution corresponding to the initialization of Stochastic Sinkhorn is the matrix obtained by first multiplying each column of $A$ by the corresponding entry of $c$ and then scaling the rows of the resulting matrix so they agree with $r$. This is equivalent to performing a full update step of \sinkhorn{} on the matrix $A \diag(c)$ at the beginning of this algorithm.
In simulations, this starting point is of better quality than the matrix $A/\|A\|_1$ which \greedy{} uses as its first iterate; however, this advantage quickly disappears. Since our goal is to compare \greedy{} and Stochastic Sinkhorn in terms of the number of required row or column updates, we also initialize \greedy{} at this point instead of at $A/\|A\|_1$ to facilitate an apples-to-apples comparison.

To compare the performance of \greedy{} with Stochastic Sinkhorn, we use an experiment on random images with $20$\% foreground pixels, as in Section~\ref{subsec:rand-images}. We initialize both algorithms with the same primal solution and used Algorithm~\ref{alg:round} to round iterates of each algorithm to the feasible polytope $\U$. Implementing Stochastic Sinkhorn requires choosing a step size, denoted by $C$ in~\cite{GenCutPey16}. That paper suggests choosing $C = 1/(Ln)$, $3/(Ln)$, or $5/(Ln)$, where $L$ is an upper bound on the Lipschitz constant of the semi-dual problem they consider.%
\footnote{In fact, they propose the step sizes $C = 1/L, 3/L, 5/L$ in the main text, but the extra factor of $n$ is present in the simulation code posted online, so we have opted to retain it in our experiments. Our experimental results indicate that without the factor of $n$, the resulting algorithm is quite unstable.}
We compare all three choices of step size with our implementation of the \greedy{} algorithm in Figure~\ref{fig:greedy_vs_stoch} with two different values of the parameter $\eta$.

\begin{figure}[ht]
\centering
\includegraphics[width=.5\textwidth]{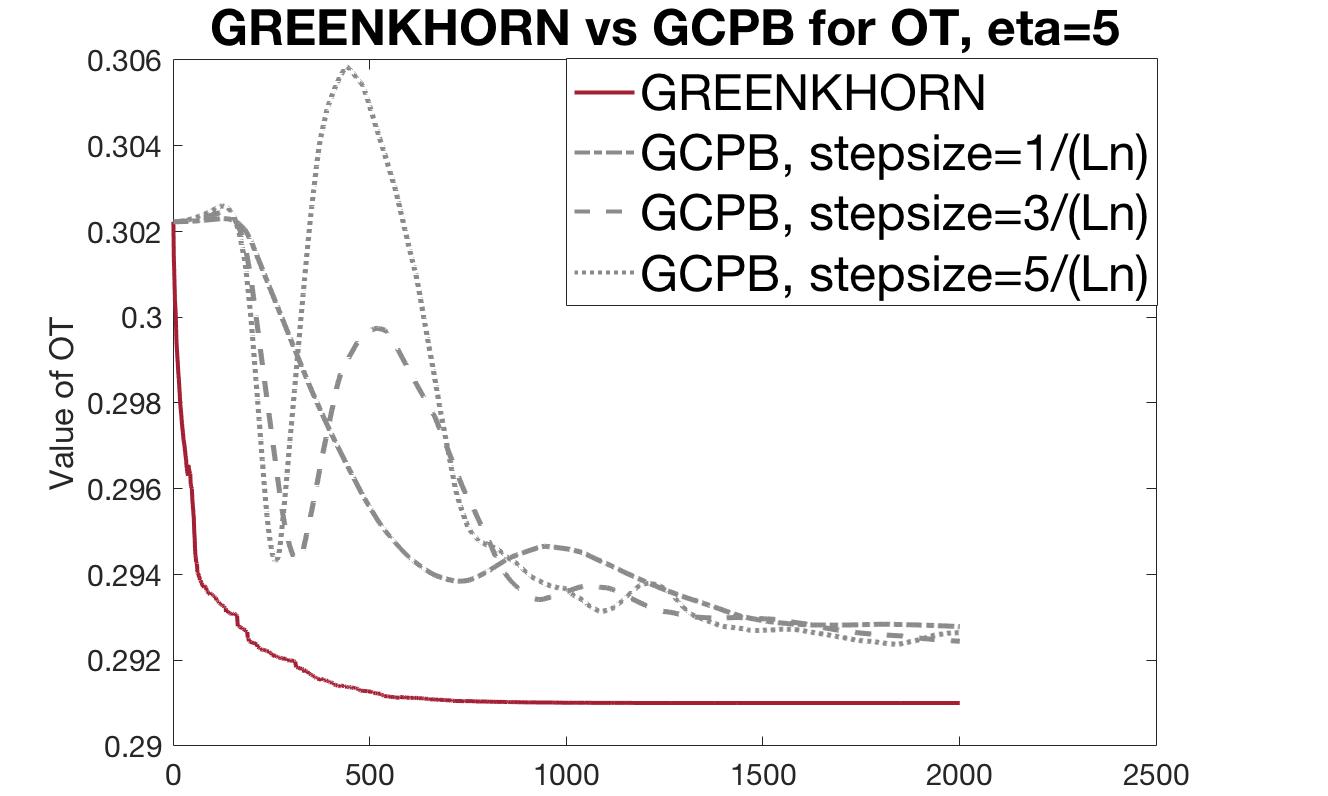}\hfill
\includegraphics[width=.5\textwidth]{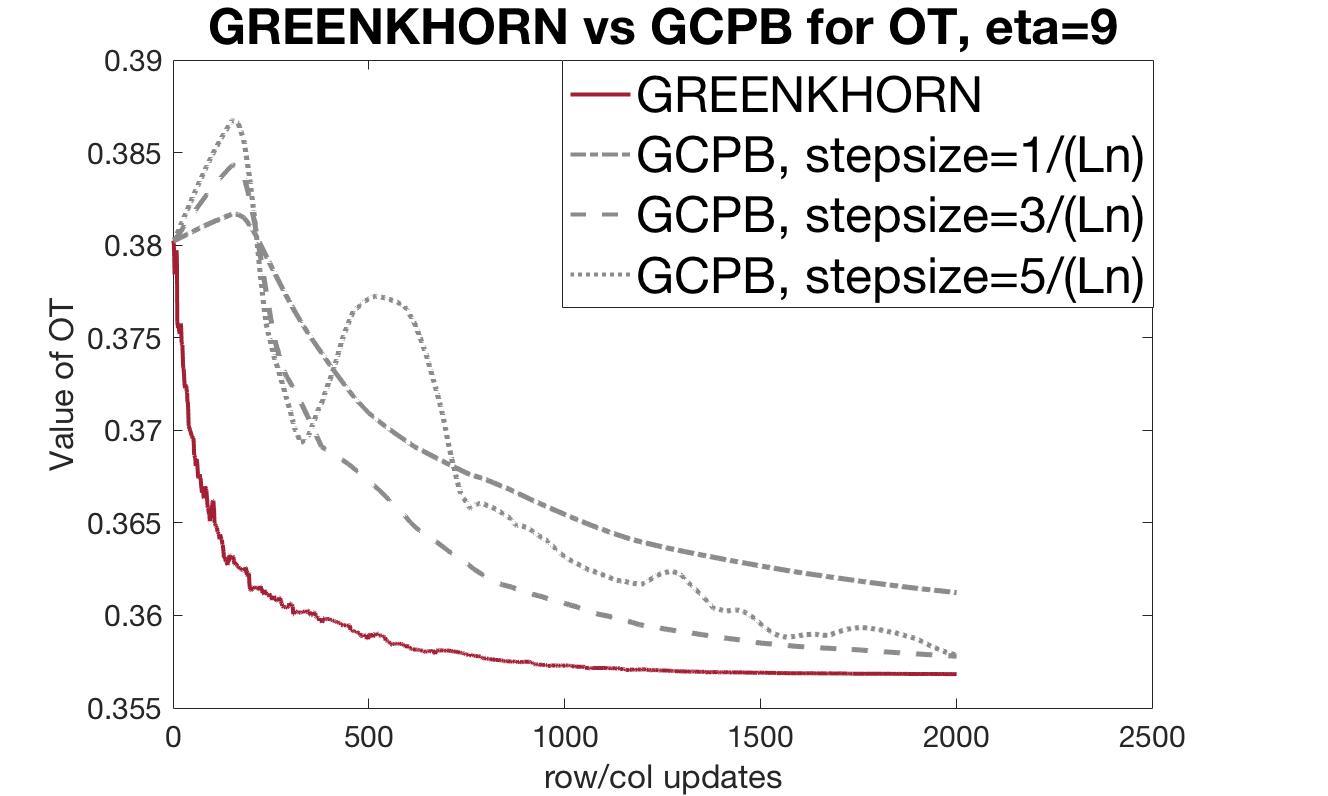}
\caption{Comparison of \greedy{} and Stochastic Sinkhorn}
\label{fig:greedy_vs_stoch}
\end{figure}

\subsection*{Acknowledgments}
We thank Michael Cohen, Adrian Vladu, Jon Kelner, Justin Solomon, and Marco Cuturi for helpful discussions. We are grateful to Pablo Parrilo for drawing our attention to the fact that \greedy{} is a coordinate descent algorithm, and to Alexandr Andoni and Inderjit Dhillon for references.

\small
\bibliographystyle{alphaabbr}
\bibliography{AltWeeRig17}
\appendix

\end{document}